\documentclass[11pt]{article}

\usepackage{amsmath,amssymb,amsfonts,bbm,bm}
\usepackage{slashed}
\usepackage{braket}
\usepackage{tocloft} 

\usepackage[nosort]{cite} 
\usepackage{graphicx,color}
\usepackage[colorlinks=true, linkcolor=blue, citecolor=blue, linktoc=all]{hyperref}

\usepackage[usenames,dvipsnames]{xcolor}

\usepackage{tikz-cd}
\usepackage{pict2e}
\usepackage{physics}
\usepackage{comment} 
\usepackage[makeroom]{cancel}

\newcommand{\beq}{\begin{eqnarray}}
\newcommand{\eeq}{\end{eqnarray}}
\newcommand{\beqn}{\begin{eqnarray}}
\newcommand{\eeqn}{\end{eqnarray}}
\newcommand{\bea}{\begin{eqnarray}}
\newcommand{\eea}{\end{eqnarray}}
\newcommand{\be}{\begin{equation}}
\newcommand{\ee}{\end{equation}}

\usepackage[normalem]{ulem}

\newcommand{\uiuc}[1]{
	\centerline{
		\begin{minipage}[c]{0.7\textwidth}
			\begin{center}
			${}^{#1}$ Illinois Center for Advanced Studies of the Universe \& Department of Physics,\\ 
			University of Illinois, 1110 West Green St., Urbana IL 61801, U.S.A.
			\end{center}
		\end{minipage}
		}
	}
   \newcommand{\nyu}[1]{
	\centerline{
		\begin{minipage}[c]{0.7\textwidth}
			\begin{center}
			${}^{#1}$ Center for Cosmology and Particle Physics, New York University, New York, NY 10003, USA
			\end{center}
		\end{minipage}
		}
	}

\usepackage{mathrsfs,mathtools}
\renewcommand\mathbb[1]{\mathbbm{#1}}

\makeatletter
\DeclareRobustCommand{\loplus}{\mathbin{\mathpalette\dog@lsemi{+}}}
\DeclareRobustCommand{\lotimes}{\mathbin{\mathpalette\dog@lsemi{\times}}}
\DeclareRobustCommand{\roplus}{\mathbin{\mathpalette\dog@rsemi{+}}}
\DeclareRobustCommand{\rotimes}{\mathbin{\mathpalette\dog@rsemi{\times}}}

\newcommand{\dog@rsemi}[2]{\dog@semi{#1}{#2}{-90,90}}
\newcommand{\dog@lsemi}[2]{\dog@semi{#1}{#2}{270,90}}
\newcommand{\dog@semi}[3]{%
  \begingroup
  \sbox\z@{$\m@th#1#2$}%
  \setlength{\unitlength}{\dimexpr\ht\z@+\dp\z@\relax}%
  \makebox[\wd\z@]{\raisebox{-\dp\z@}{%
    \begin{picture}(1,1)
    \linethickness{\variable@rule{#1}}
    \roundcap
    \put(0.5,0.5){\makebox(0,0){\raisebox{\dp\z@}{$\m@th#1#2$}}}
    \put(0.5,0.5){\arc[#3]{0.5}}
    \end{picture}%
  }}%
  \endgroup
}
\newcommand{\variable@rule}[1]{%
  \fontdimen8  
  \ifx#1\displaystyle\textfont3\else
    \ifx#1\textstyle\textfont3\else
      \ifx#1\scriptstyle\scriptfont3\else
        \scriptscriptfont3\relax
  \fi\fi\fi
}
\DeclareRobustCommand{\loplus}{\mathbin{\mathpalette\dog@lsemi{+}}}

\newcommand{\Hom}[2]{\text{Hom}(#1,#2)}
\usepackage{float}
\newtheorem{theorem}{Theorem}[section]
\newtheorem{proof}{Proof}[section]

\newtheorem{definition}{Definition}[section]

\usepackage{enumitem}
\newlist{encase}{enumerate}{1}
\setlist[encase]{label=Case \arabic*:, leftmargin=2cm}

\usepackage{tikz-cd}
\newcommand{\thistitle}{Extensions from within} 
 
\begin{document}

\title{\thistitle}
\author{
Shadi Ali Ahmad$^{a}$ and Marc S. Klinger$^{b}$
	\\ 
	\\
    {\small \emph{\nyu{a}}} \\ \\
	{\small \emph{\uiuc{b}}}
	\\
	}
\date{}
\maketitle
\vspace{-0.5cm}
\begin{abstract}
Inclusions and extensions lie at the heart of physics and mathematics. The most relevant kind of inclusion in quantum systems is that of a von Neumann subalgebra, which is the focus of this work. We propose an object intrinsic to a given algebra that indexes its potential extensions into larger algebras. We refer to this object as a spatial Q-system, since it is inspired by the Q-system construction of Longo and others that has been used to categorify finite index inclusions for properly infinite algebras. The spatial Q-system broadens the usefulness of the Q-system by extending its applicability to the context of possibly infinite index inclusions admitting only operator valued weights rather than conditional expectations. An immediate physical motivation for the spatial Q-system arises in the form of the crossed product, which is a method for extending a given algebra by the generators of a locally compact group which acts upon it. If the group in question is not finite, the inclusion of the original algebra into its crossed product will be of infinite index and admit only an operator valued weight. Such an extension cannot be described by a standard Q-system. The spatial Q-system covers this case and, in fact, may be regarded as a far reaching generalization of the crossed product in which an algebra is extended by a collection of operators that possess a closed product structure compatible with the original algebra, but need not generate a symmetry object. We comment upon the categorical interpretation of the spatial Q-system, and describe how it may be interpreted as instantiating a protocol for approximate, or non-isometric, quantum error correction or in terms of quantum reference frames. We conclude by identifying a series of physical applications including the study of generalized symmetries, non-perturbative quantum gravity, black hole evaporation and interiors, and holography.
\vspace{0.3cm}
\end{abstract}
\begingroup
\hypersetup{linkcolor=black}
\tableofcontents
\endgroup
\section{Introduction} 

Two fundamental notions that emerge in various guises throughout physics and mathematics are inclusions and extensions. A standard notion of inclusion that appears in physics is that of a subsystem. For example, a quantum field theory may be formalized in terms of a net of operator algebras, which is an assignment of observables to given subregions of spacetime. When the theory in question admits a gauge symmetry it is necessary to further restrict observable degrees of freedom to satisfy symmetry constraints. In this way, the act of gauging a theory may be regarded as giving rise to an inclusion of the gauge invariant operators into the full kinematical algebra. We can turn this question on its head. Given a net of gauge invariant algebras, what kinematical algebras exist which admit the former as a neutral sector? From this point of view, we regard the kinematical system as an extension of its invariant subalgebra. This question was answered by Doplicher and Roberts in seminal work \cite{Doplicher:1972kr,Doplicher1989}, where they demonstrated that one can completely recover the structure of the symmetry and its action on the kinematical net provided it is encoded in a compact group.

A construction which is closely related to Doplicher and Roberts' result is the crossed product of a von Neumann algebra by a locally compact group \cite{10.1007/BF02392041}. In essence, the crossed product is an approach to \emph{extending} a von Neumann algebra by the generators of a group which acts automorphically on it. The crossed product has received a large amount of attention in recent work which has explored its usefulness in a wide range of physical applications. The modular crossed product, for example, has appeared as a natural device for encoding subregion algebras in perturbative quantum gravity \cite{Witten:2021unn,Chandrasekaran:2022cip,Chandrasekaran:2022eqq,Jensen:2023yxy,AliAhmad:2023etg,Klinger:2023tgi,Kudler-Flam:2023hkl,Kudler-Flam:2023qfl,AliAhmad:2024eun}. Remarkably, extending subregion algebras from quantum field theory by incorporating generators of the modular automorphism smooths divergences that would otherwise appear in the computation of entropic observables like the entanglement entropy. In this regard, the modular crossed product assumes the role of a regulator of QFT subregion physics. At the same time, the crossed product has also emerged as a tool for constructing invariant algebras in gauge theory \cite{Chandrasekaran:2022cip,Jensen:2023yxy,Klinger:2023auu}. In this capacity, the extended degrees of freedom are used to dress operators in such a way as to satisfy grouplike symmetry constraints. Finally, the crossed product can also be understood as an algebraic formalization of quantum reference frames \cite{AliAhmad:2024wja,AliAhmad:2024vdw,Fewster:2024pur,DeVuyst:2024pop}, with the extended degrees of freedom playing the role of probes which define a measurement protocol relative to a chosen observer. 

Despite its usefulness, the crossed product is not without its limitations. In its conventional form, the crossed product only allows for the extension of an algebra by the generators of a locally compact group. Generalizations of the crossed product also exist for quantum groups \cite{enock1977produit,enock1980produit} and locally compact groupoids \cite{renault2006groupoid,williams2007crossed,muhly2008renaultsequivalencetheoremgroupoid}. However, one might be interested in understanding extensions which do not rely upon the existence of a grouplike structure underlying them. Given a von Neumann algebra $M$ in which another algebra $N$ is embedded as a subalgebra $N \subset M$, we unlock a dual perspective. We may regard (1) $N$ as an \textit{inclusion} in $M$ and/or (2) $M$ as an \textit{extension} of $N$. In this work, we present an approach to algebraic extensions which substantially generalizes the crossed product while retaining many of its desirable features. The question is no longer for which algebra $M$ is $N$ the neutral part, but generally where does $N$ fit compatibly.  

Our construction builds upon seminal work by Longo and others introducing and expounding upon a structure called a Q-system~\cite{cmp/1104254494,Longo:1989tt,Bischoff:2014xea}. We give a lengthy introduction to the Q-system in the main body of the note, so for now let us simply say that it is a method for constructing inclusions of operator algebras $N \subset M$ utilizing only features which are intrinsically defined with respect to the smaller algebra $N$. The central ingredient in the specification of a Q-system is an endomorphism of $N$, $\theta: N \rightarrow N$. Given a fixed algebra $N$, each Q-system, $A$, based around a given endomorphism, $\theta \in \text{End}(N)$, defines a different extension of $N$ which we denote by $M_A \equiv N \times_{\theta} A$. Famously, the crossed product of a von Neumann algebra by a finite group can be shown to be a special kind of Q-system extension. Thus, the Q-system presents itself as a tool for constructing algebraic extensions which includes crossed products as a subcase. 

Unfortunately, in their original form Q-systems can only be used to construct extensions with finite index. The index of an inclusion is sometimes qualitatively described as a measure of how many times a subalgebra `fits inside' a larger algebra in which it is embedded. When working with inclusions of properly infinite algebras, which appear under general considerations in quantum field theory and holography, it is somewhat rare to find inclusions of finite index. Moreover, the crossed product of a von Neumann algebra by a locally compact group which is not finite will always give rise to an inclusion of infinite index. Finally, the existence of finite index inclusions is intimately related to the existence of a special class of maps on von Neumann algebras called conditional expectations. In fact, the existence or non-existence of conditional expectations is a key consideration in the theory of quantum error correction. Inclusions of infinite index may not admit conditional expectations, which appears to imply that they may not allow for exact quantum error correction. Approximate error correction is an important subject in quantum information theory and holography \cite{Almheiri:2014lwa,Pastawski:2015qua,Harlow:2016vwg,Faulkner:2020hzi,Faulkner:2020iou,Faulkner:2020kit,Faulkner:2022ada,Akers:2022qdl}. With these observations in mind, it is of great interest to establish a generalization of Longo's Q-system which can apply to infinite index inclusions. In this note, we obtain such an object.

It is also known that the quantum error correction nature of holography gives rise to a generalized entropy formula for the bulk encoded in the boundary. In related recent work~\cite{AliAhmad:2024saq}, we capitalize on this observation to derive a generalized entropy formula for inclusions of von Neumann algebras admitting only a generalized conditional expectation. It is seen that the relaxation of a conditional expectation down to a generalized one frees the area operator from centrality, making it more consistent with expectations from quantum gravity and the backreaction of matter on spacetime. The philosophy of this note, namely that of extending an algebra by some non-local complex operators is infused in our result on generalized entropy as one can think of the encoding of the bulk into the boundary as a \textit{spatial} Q-system which is not one of finite index.  

To center our discussion, we formulate our presentation around the problem of constructing Q-systems which describe crossed product algebras. As we have addressed, the standard Q-system is sufficient for describing the crossed product of a von Neumann algebra by a finite group, but this leaves open the question of extending this structure to crossed products by infinite groups. More broadly, we divide our work into three cases, here $N \subset M$ is an inclusion of operator algebras:
\begin{enumerate}
    \item There exists a conditional expectation of finite index between $M$ and $N$,
    \item There exists a conditional expectation between $M$ and $N$ of possibly infinite index,
    \item There need not exist conditional expectations between $M$ and $N$. 
\end{enumerate}
The first case coincides with Longo's construction which we shall henceforward refer to as the standard Q-system. In recent work \cite{Del_Vecchio_2018}, a relaxation of the standard Q-system was proposed, called the generalized Q-system, which covers the second case. The main contribution of this work is to introduce a new notion called the \textit{spatial Q-system} which subsumes all three cases.

The paper is organized as follows. In Section \ref{i ovw ce}, we sharply define the cases of inclusions we will be interested in defining Q-systems for. In Section \ref{sec: QCE} we systematically derive the definition of the generalized Q-system which characterizes inclusions of operator algebras admitting conditional expectations of possibly infinite index. We explore the crossed product by a finite group as a paradigmatic example of this structure. In Section \ref{sec: Longo} we introduce the standard Q-system as a special case of the generalized Q-system. In Section \ref{sec: QOVW} we repeat the analysis of the previous section when the conditional expectation has been replaced with a non-unital operator valued weight. We demonstrate that the features of the Q-system are retained in a relaxed form even in this case leading to the notion of a Spatial Q-system. In Section \ref{sec: spatial recon}, we prove a reconstruction theorem establishing how one can build an extension algebra from the data of a Spatial Q-system. In Section \ref{sec: comp} we provide a comparison between the distinct notions of Q-system which coincide with the three cases introduced above. In particular, we describe how these constructions differ on a mathematical level in terms of their categorical structure, and on a physical level in terms of the level of quantum error correction they admit or the symmetry objects the theory can entertain. We conclude in Section \ref{sec: disc} in which we identify several physical settings that require a theory of infinite index inclusions and therefore constitute interesting directions for future work. There are two appendices. The first, Appendix~\ref{app: prelim}, where  we review relevant background from von Neumann algebras including the modular theory of weights, and index theory. We also provide context for the Q-system by reviewing the crossed product as a grouplike extension of a von Neumann algebra. The second, Appendix~\ref{app: Cat vN}, briefly reviews the categorification of structures related to von Neumann algebras and maps between them.

\section{Inclusions, Operator Valued Weights, and Conditional Expectations}\label{i ovw ce}

The main structure we are interested in studying in this note is the inclusion of von Neumann algebras $i: N \hookrightarrow M$.\footnote{For our purposes, all inclusions are taken to be unital.} The tool we will utilize for understanding the statistical embedding of $N$ inside of $M$ is the notion of an operator valued weight \cite{haagerup1979operator,haagerup1979operator2}. Recall that an operator valued weight is a map $T: M_+ \rightarrow \hat{N}_+$, where here $\hat{N}$ is the algebra of operators \emph{affiliated}\footnote{Given a representation $\pi: N \rightarrow B(H)$, a linear map $\mathcal{O}: H \rightarrow H$ is affiliated with $N$ if it commutes with all operators in the commutant of $N$. If, moreover, $\mathcal{O}$ is a bounded operator, then it belongs to $N$.} with $N$, satisfying the properties
\begin{enumerate}
	\item $T(\sum_{i} \lambda_i m_i) = \sum_{i} \lambda_i T(m_i), \qquad \forall \lambda_i \in \mathbb{R}_+, \; m_i \in M$,
	\item $T(i(n)^* \; m \; i(n)) = n^* \; T(m) \; n, \qquad \forall m \in M, \; n \in N$. 
\end{enumerate}

The `domain' of an operator valued weight is defined
\beq
	\mathfrak{n}_{T} \equiv \{m \in M \; | \; \norm{T(m^*m)} < \infty\},
\eeq
and we denote by $\mathfrak{m}_T \equiv \mathfrak{n}_T^* \mathfrak{n}_T$ the span of products $m_1^* m_2$ for $m_1,m_2 \in \mathfrak{n}_T$. When restricted to $\mathfrak{m}_T$, the operator valued weight can be regarded as a map $T: \mathfrak{m}_T \rightarrow N$ which satisfies the homogeneity property
\beq \label{OVW Homogeneity}
	T(i(n_1) \; m \; i(n_2)) = n_1 \; T(m) \; n_2, \qquad n_1,n_2 \in N, \; m \in \mathfrak{m}_T. 
\eeq
By a slight abuse of notation, we will often regard $T$ as a map from $M$ to $\hat{N}$ satisfying \eqref{OVW Homogeneity}, although strictly speaking one should carefully consider the domains involved. We say that $T$ is (i) \emph{faithful} if $T(m^*m) = 0 \iff m = 0$, (ii) \emph{semifinite} if $\mathfrak{n}_T$ is weakly dense in $M$, and (iii) \emph{normal} if the sequence $\{T(m_i)\}_{i \in \mathcal{I}}$ converges to $T(m)$ in $\hat{N}_+$ when the sequence $\{m_i\}_{i \in \mathcal{I}}$ converges to $m$ in $M_+$. We denote the space of faithful, semifinite, normal operator valued weights associated with the inclusion $i: N \hookrightarrow M$ by $P(M,N)$. 

In the event that $T(\mathbb{1}_M) = \mathbb{1}_N$, the operator valued weight may be regarded as a norm one projection of vector spaces. Such a norm one projection which satisfies the homogeneity property of an operator valued weight is called a conditional expectation. We denote the space of conditional expectations from $M$ to $N$ by $C(M,N) \subset P(M,N)$.\footnote{Here we are restricting our analysis to faithful, semifinite, normal conditional expectations although one could consider more general cases.} In the event that $N = \mathbb{C}$, an operator valued weight $T \in P(M,\mathbb{C})$ is nothing but an ordinary weight. A conditional expectation $\varphi \in C(M,\mathbb{C})$ is called a state. For ease of notation we will often write $P(M) \equiv P(M,\mathbb{C})$, and $C(M) \equiv C(M,\mathbb{C})$. 

In this note, we organize our presentation as a sequence of increasing abstraction for an inclusion $N\subset M$
\begin{encase}
	\item The sets $C(M,N)$ and $C(\pi(N)',\pi(M)')$ are both non-empty. 
	\item Either $C(M,N)$ or $C(\pi(N)',\pi(M)')$ is non-empty.
	\item $i: N \hookrightarrow M$ is an inclusion for which the set $P(M,N)$ is non-empty. 
\end{encase}
Clearly, Case 1 $\subset$ Case 2 $\subset$ Case 3. The crossed product by a finite group fits into Case 1, while the crossed product by an infinite locally compact group fits into Case 3. Case 2 is a useful intermediate case that will be used to show how some of the structure in Case 1 can be relaxed. 

\section{Q-Systems for Conditional Expectations} \label{sec: QCE}

In this section we cover Case 1 and Case 2. Given an inclusion $i: N \hookrightarrow M$, we first demonstrate how the existence of a conditional expectation $E: M \rightarrow N$ endows $M$ with the structure of a module with $N$-valued inner product. This inner product allows us to regard $M$ as being generated by elements of $N$ along with a basis of elements in $M$. Working in the GNS representation of $M$, the existence of the conditional expectation can moreover be brought into correspondence with the existence of an intertwining isometry $W$ that implements $E$ spatially. This map is naturally related to the Jones projection. In the event that $E$ has finite index, we can also construct a second isometry $\overline{W}$ implementing a dual conditional expectation on the commutant inclusion. This is the content of Jones' basic construction. The pair of isometries $W$ and $\overline{W}$ form a conjugate system and allow us to recover the standard Q-system as introduced by Longo. We explore the categorification of this notion and demonstrate how it allows us to understand the algebra $M$ as emerging entirely from data contained in $N$. 

Let $i: N \hookrightarrow M$ be an inclusion of von Neumann algebras admitting a conditional expectation $E: M \rightarrow N$.\footnote{For simplicity, we will work in the context that $C(M,N)$ is non-empty while $C(\pi(N)',\pi(M)')$ is possibly empty. Differences between this case and the alternative are discussed in depth in \cite{Del_Vecchio_2018}.} Given an $E$-invariant state, $\varphi \in C(M)$, it is not difficult to show that $\varphi = \varphi_0 \circ E$ where $\varphi_0 \equiv \varphi \circ i \in C(N)$. We denote by $L^2(M;\varphi)$ the GNS Hilbert space of $M$ with respect to $\varphi$, and by $L^2(N;\varphi_0)$ the GNS Hilbert space of $N$ with respect to $\varphi_0$. In fact, $L^2(N;\varphi_0) \subset L^2(M;\varphi)$ as Hilbert spaces with with $L^2(N;\varphi_0)$ regarded as the set of vectors in $L^2(M;\varphi)$ which are generated by acting on $\xi_{\varphi}$ with operators $\pi_{\varphi} \circ i(n)$ for each $n \in N$. 

The conditional expectation $E$ allows us to define an $N$-valued bilinear in $M$ by
\beq
	G_E: M^{\times 2} \rightarrow N, \qquad G_E(m_1,m_2) \equiv E(m_1^* m_2). 
\eeq
By a generalization of the Gram-Schmidt procedure, the existence of the $N$-valued inner product $G_E$ implies the existence of a basis $\{\lambda_i\}_{i \in \mathcal{I}} \subset M$ such that any element $m \in M$ admits a collection of $N$-valued coefficients $\{m_i\}_{i \in \mathcal{I}} \subset N$ such that
\beq \label{Resolution by Pimsner-Popa}
	m = \sum_{i \in \mathcal{I}} \lambda_i^* \; i(m_i).
\eeq
In particular, the coefficients $m_i$ are obtained by taking the $N$-valued inner product of $m$ with each basis element, $m_i = G_E(\lambda_i^*,m)$. The collection $\{\lambda_i\}_{i \in \mathcal{I}}$ was first recognized by Pimsner and Popa \cite{pimsner1986entropy} and is subsequently referred to as a Pimsner-Popa basis. In fact, a Pimsner-Popa basis can be defined more axiomatically as a collection of elements $\{\lambda_i\}_{i \in \mathcal{I}} \subset M$ satisfying two properties:
\begin{enumerate}
	\item $E(\lambda_i \lambda_j^*) = q_i \delta_{ij}$ where $q_i \equiv E(\lambda_i \lambda_i^*)$ is a projection in $N$,
	\item $L^2(M) = \sum_{i \in \mathcal{I}} \pi(\lambda_i^*) L^2(N;\varphi_0)$.  
\end{enumerate}
The second of these properties is equivalent to Eq.~\eqref{Resolution by Pimsner-Popa}. We conclude that the full algebra $M$ may be regarded as the von Neumann union $M \simeq i(N) \vee \{\lambda_i\}_{i \in \mathcal{I}}$. 

The cardinality of the set $\{\lambda_i\}_{i \in \mathcal{I}}$ is not an invariant, and in general there can be many Pimsner-Popa bases associated with a given conditional expectation $E$. With this being said, the operator
\beq
	\text{Ind}(E) \equiv \sum_{i \in \mathcal{I}} \lambda_i^* \lambda_i
\eeq
is an invariant and defines a general notion of the index of the conditional expectation $E$ even when $M$ and $N$ are not factors. In general, $\text{Ind}(E)$ takes values in the algebra affiliated with the center of $M$. We say the index is finite if it is a bounded operator, in which case it is a central element of $M$. A useful result which we will make use of later on is that, if $N$ is a properly infinite algebra, the index will be finite if and only if there exists a Pimsner-Popa basis of cardinality one \cite{Del_Vecchio_2018}. 

Since $\varphi$ and $\varphi_0$ are states, $\mathbb{1}$ belongs to their domains. We can therefore form Connes' spatial operator as in \eqref{Connes Spatial Operator}:
\beq
	W \equiv K^{\varphi}_{\eta_{\varphi}(\mathbb{1})}: L^2(N;\varphi_0) \rightarrow L^2(M;\varphi), \qquad W(\eta_{\varphi_0}(n)) = \pi_{\varphi} \circ i(n) \eta_{\varphi}(\mathbb{1}) = \eta_{\varphi} \circ i(n). 
\eeq	
Here, we have used the fact that $L^2(M;\varphi)$ is a representation space for the algebra $N$ via the representation $\pi_{\varphi} \circ i: N \rightarrow B(L^2(M;\varphi))$. By the properties of the spatial operator, $W \in \text{Hom}^0_N(\pi_{\varphi_0},\pi_{\varphi} \circ i)$. The formal adjoint, $W^{\dagger} \in \text{Hom}^0_{N}(\pi_{\varphi} \circ i, \pi_{\varphi_0})$ is defined by the equation
\beq \label{Adjoint intertwiner}
	g_{L^2(M;\varphi)}\bigg(W(\eta_{\varphi_0}(n)), \eta_{\varphi}(m)\bigg) = g_{L^2(N;\varphi_0)}(\eta_{\varphi_0}(n), W^{\dagger}(\eta_{\varphi}(m))\bigg), \qquad \forall n \in \mathfrak{n}_{\varphi_0}, m \in \mathfrak{n}_{\varphi}. 
\eeq
Expanding Eq.~\eqref{Adjoint intertwiner}, we can write
\beq \label{Adjoint intertwiner 2}
	\varphi_0 \circ E(i(n)^* m) = \varphi_0(n^* E(m)) = \varphi_0\bigg(n^* \eta_{\varphi_0}^{-1}(W^{\dagger}(\eta_{\varphi}(m))\bigg).
\eeq
To move from the first to the second equality we have used the homogeneity property of the conditional expectation. From Eq.~\eqref{Adjoint intertwiner 2} we conclude that
\beq \label{Adjoint of W}
	W^{\dagger}(\eta_{\varphi}(m)) = \eta_{\varphi_0} \circ E(m). 
\eeq

By construction, the map $W$ implements the conditional expectation $E$ in the following sense. For any $m \in M$ and $n \in \mathfrak{n}_{\varphi_0}$ we have
\begin{flalign}
	W^{\dagger} \pi_{\varphi}(m)  W(\eta_{\varphi_0}(n)) &= W^{\dagger} \pi_{\varphi}(m) \eta_{\varphi} \circ i(n) \nonumber \\
	&= W^{\dagger} \eta_{\varphi}(m i(n)) = \eta_{\varphi_0} \circ E(m i(n)) = \pi_{\varphi_0} \circ E(m) \eta_{\varphi_0}(n). 
\end{flalign}
In other words,
\beq
	E(m) = \pi_{\varphi_0}^{-1} \circ \text{Ad}_{W^{\dagger}} \circ \pi_{\varphi}(m), \qquad \forall m \in M. 
\eeq
Formally, this implies that $W: L^2(N;\varphi_0) \rightarrow L^2(M;\varphi)$ is a spatial implementation of the conditional expectation $E$ when viewed as a completely positive map. The fact that $E$ is a unital map implies that $W$ is an isometry:
\beq
	\mathbb{1} = W^{\dagger} \pi_{\varphi}(\mathbb{1}) W = W^{\dagger} W. 
\eeq

On the other hand, the composition $P \equiv W W^{\dagger}: L^2(M;\varphi) \rightarrow L^2(M;\varphi)$ is an orthogonal projection map. Since $W \in \text{Hom}^0_N(\pi_{\varphi_0},\pi_{\varphi} \circ i)$ and $W^{\dagger} \in \text{Hom}^0_N(\pi_{\varphi} \circ i, \pi_{\varphi_0})$ the projection $P$ belongs to $\text{Hom}^0_N(\pi_{\varphi} \circ i, \pi_{\varphi} \circ i) = \pi_{\varphi} \circ i(N)'$. This is called the Jones projection. Its not difficult to show that
\beq
	P \pi_{\varphi}(m) P = \pi_{\varphi} \circ i \circ E(m) P, \qquad \forall m \in M. 
\eeq
In this respect, the Jones projection implements the conditional expectation entirely within the standard representation $L^2(M;\varphi)$. Using the projection $P$ we define the Jones basic extension $M_1 \equiv M \vee \{P\}$, where here $\{P\}$ is the algebra generated by $P$. Clearly, $M \subset M_1$ and, in fact, the inclusion $i_1: M \hookrightarrow M_1$ is anti-isomorphic to the commutant inclusion $i': \pi_{\varphi}(M)' \hookrightarrow \pi_{\varphi} \circ i(N)'$ \cite{kosaki1998type}. 

To summarize, in the above discussion we have shown that the existence of a conditional expectation $E: M \rightarrow N$ is equivalent to the existence of an endomorphism $\theta \equiv \pi_{\varphi} \circ i: N \rightarrow B(L^2(M;\varphi))$, a Pimsner-Popa basis $\{\lambda_i\}_{i \in \mathcal{I}} \subset M$, and an isometric intertwiner $W \in \text{Hom}^0_{N}(\pi_{\varphi_0}, \theta)$. In the sequel, we shall see how $(\theta, \{\pi_{\varphi}(\lambda_i)\}_{i \in \mathcal{I}}, W)$ is an example of a construction we refer to as a \emph{Spatial Q-system}. We will also demonstrate that each such triple \emph{defines} an algebraic extension complete with a conditional expectation (or more generally, an operator valued weight). For now, let us pause to provide a useful example of our construction in the form of the crossed product by a finite group. 

\subsection{Crossed Products by Finite Groups}

Let $(N,G,\alpha)$ be a von Neumann covariant system with $G$ a finite group. Denote by $M \equiv N \times_{\alpha} G$ the crossed product algebra. If $(H_G \equiv L^2(G,H), \pi_{\alpha},\lambda)$ is a covariant representation of $(N,G,\alpha)$ we may regard the crossed product as comprised of operators of the form
\beq \label{Crossed Product Expansion}
	m = \sum_{g \in G} \lambda(g) \pi_{\alpha}(m(g)).
\eeq
Comparing Eqs.~\eqref{Crossed Product Expansion} and \eqref{Resolution by Pimsner-Popa} we see that $\{\lambda(g)\}_{g \in G}$ are playing the role of a Pimsner-Popa basis. The conditional expectation $E: M \rightarrow N$ acts on elements of the form Eq.~\eqref{Crossed Product Expansion} as
\beq
	E(m) = m(e),
\eeq 
so that
\beq
	E(\lambda(h^{-1}) m) = E\bigg( \sum_{g \in G} \lambda(h^{-1} g) \pi_{\alpha}(m(g))\bigg) = \delta(h^{-1}g - e) m(g) = m(h). 
\eeq
In this case, the isometry $W$ intertwines the representations $\pi_{\alpha}$ and $\pi$ of $N$. Let $\Omega \in H_G$ be the vector representative of an $E$-invariant state $\varphi \in C(M)$. Then,
\beq
	W\bigg((\pi(n) \otimes \mathbb{1}) \Omega \bigg) = \pi_{\alpha}(n) \Omega. 
\eeq
Its formal adjoint acts as
\beq
	W^{\dagger}\bigg(\sum_{g \in G} \lambda(g) \pi_{\alpha}(m(g)) \Omega \bigg) = (\pi(m(e)) \otimes \mathbb{1}) \Omega. 
\eeq
The triple $(\pi_{\alpha}, \{\lambda(g)\}_{g \in G}, W)$ is the Spatial Q-system associated with the finite crossed product algebra. The index of this inclusion is equal to the order of the group, as can be observed from the simple computation:
\beq
	\sum_{g \in G} \lambda(g) \lambda(g^{-1}) = \sum_{g \in G} 1 = |G|. 
\eeq

\subsection{Recasting the Generalized Q-System}

In the previous section, we demonstrated that the existence of a conditional expectation associated with an inclusion of operator algebras can be brought into correspondence with the existence of a Spatial Q-system. We would now like to discuss how this construction can be recast in a manner that traces its pieces entirely to the algebra $N$. In this sense, we can view the generalized Q-system as a device for \emph{constructing} the algebra $M$ from the intrinsic structure of $N$. 

To reach this conclusion, we must work in the context where $N$ is properly infinite. Then, we can make use of the following theorem. For a proof, we refer the reader to \cite{Del_Vecchio_2018}. We should also mention that the seeds of this approach were first planted by Fidaleo and Isola in \cite{fidaleo1995conjugate,fidaleo1999canonical}.  
\begin{theorem} \label{Generalized Q-system thm}
Let $N$ be a properly infinite von Neumann algebra and $\theta: N \rightarrow N$ an endomorphism. Then, the following are equivalent:
\begin{enumerate}
	\item There exist von Neumann algebras $N_2 \subset N_1 \subset N$ such that $C(N_1,N_2)$ is non-empty, and $\theta$ is the canonical endomorphism of the inclusion $N_1 \subset N$.
	\item There exists a von Neumann algebra $M$ such that $N \subset M$ with canonical endmorphism $\gamma \in \text{End}(M)$, such that $C(M,N)$ is non-empty and $\theta = \gamma\rvert_{N}$. 
\end{enumerate}
Moreover, the canonical endomorphism $\gamma$ intertwines conditional expectations $\tilde{E} \in C(N_1, N_2)$ and $E \in C(M,N)$ such that $\tilde{E} \circ \gamma = \gamma \circ E$. This defines a bijection between $C(N_1,N_2)$ and $C(M,N)$. 
\end{theorem}

The above theorem tells us that we can reformulate the Spatial Q-system associated with a conditional expectation $E \in C(M,N)$ in terms of a corresponding conditional expectation $\tilde{E} \in C(N_1,N_2)$. In this case, the Q-system is comprised of (i) an endomorphism $\theta \in \text{End}(N)$, (ii) a Pimsner-Popa basis $\{\nu_i\}_{i \in \mathcal{I}} \subset N$, and an isometric intertwiner $\tilde{W}$ which implements $\tilde{E}$. Since $\tilde{W}$ can be taken to act entirely in the GNS Hilbert space of $N$, it can be identified with an element $\tilde{w} \in N$ such that $\pi_{\varphi_0}(\tilde{w}) = \tilde{W}$. Its operator product with the rest of the algebra is dictated by the fact that it is regarded as an intertwiner of \emph{endomorphisms}:
\beq
	\tilde{w} \in \text{Hom}^0_{\text{End}(N)}(\theta, id) \implies \tilde{w} n = \theta(n) \tilde{w}, \; \forall n \in N. 
\eeq
In fact, this is equivalent to the observation that the GNS representation $\pi_{\varphi_0}$ plays the role of the identity in the category of modules of $N$. We will revisit this point in the sequel. 

The collection $(\theta, \{\nu_i\}_{i \in \mathcal{I}}, \tilde{w})$ defines what is referred to as a generalized Q-system for the algebra $N$. If $N$ is properly infinite, theorem \ref{Generalized Q-system thm} ensures that the generalized Q-system gives rise to a dual Spatial Q-system $(\theta, \{\lambda_i\}_{i \in \mathcal{I}}, W)$ which identifies a conditional expectation $E \in C(M,N)$. 

\subsection{Standard Q-systems} \label{sec: Longo}

The previous subsection provides a complete spatial classification of inclusions admitting conditional expectations. We would now like to consider a special case in which the conditional expectation in question possesses a finite index. In this context, the generalized Q-system reduces to the standard Q-system originally introduced by Longo \cite{cmp/1104254494}. The standard Q-system possesses a very compelling categorical structure which we describe in detail. The typical presentation of this subject is rather technical, so we have taken care to include an extended discussion which explains why the Q-system is capable of producing extensions of von Neumann algebras \emph{from within}.

Let $i: N \hookrightarrow M$ be an inclusion of properly infinite von Neumann algebras, and $E \in C(M,N)$ a conditional expectation with finite index. As alluded to above, in this case one can construct a Pimsner-Popa basis consisting of a single element. As a result, the generalized Q-system dual to $E$ is a triple $(\theta, x, w)$ consisting of an endomorphism $\theta \in \text{End}(N)$, an isometric intertwiner $w \in \text{Hom}^0_{\text{End}(N)}(id,\theta)$ and an element $x \in N$. The element $x$ can be chosen to be an intertwiner $x \in \text{Hom}^0_{\text{End}(N)}(\theta,\theta^2)$. That is,
\beq
	x \theta(n) = \theta^2(n) x, \qquad \forall n \in N. 
\eeq
It turns out that the triple $(\theta,x,w)$ defines a structure known as a \emph{Frobenius algebra} inside of the category $\text{End}(N)$. 

As is reviewed in Appendix \ref{app: Cat vN}, $\text{End}(N)$ is a $C^*$-tensor category with objects given by endomorphisms $\alpha: N \rightarrow N$ and arrows given by intertwiners between endomorphisms. A $C^*$ Frobenius algebra in $\text{End}(N)$ is a collection $A = (\theta,x,w)$ with $\theta: N \rightarrow N$, $w \in \text{Hom}^0_{\text{End}(N)}(id,\theta)$, and $x \in \text{Hom}^0_{\text{End}(N)}(\theta,\theta^2)$ satisfying the following three conditions:
\begin{enumerate}
	\item $w^* x = \theta(w^*) x = \mathbb{1}$,
	\item $x^2 = \theta(x) x$, and
	\item $\theta(x^*) x = x x^* = x^* \theta(x)$\footnote{It was pointed out to us by Roberto Longo that this condition is unnecessary and is implied by the first two as shown in Ref.~\cite{longo1996theorydimension}}. 
\end{enumerate}
If $x^*x$ is a multiple of $\mathbb{1}$, we say that $A$ is special. If, moreover, $w^* w = d_A \mathbb{1}_{id_N}$ and $x^* x = d_A \mathbb{1}_{\theta}$ then $A$ is called standard. In general $d_A = \sqrt{\text{dim}(\theta)}$ is called the dimension of $A$. 

The fundamental example of a Q-system which motivates its usefulness in extension theory is given as follows. Let $N$ be a von Neumann factor and $M$ a von Neumann algebra with $\alpha: N \rightarrow M$ and $\overline{\alpha}: M \rightarrow N$ a conjugate pair of homomorphisms. Associated to this pair, we have a conjugate pair of intertwiners $(w,\overline{w})$. Then $A_{\alpha} = (\overline{\alpha} \circ \alpha, \overline{\alpha}(\overline{w}), w)$ is a $C^*$ Frobenius algebra inside $\text{End}(N)$. It is automatically special and will be standard if and only if $(w,\overline{w})$ is a standard pair. In that case $d_{A_{\alpha}} = \sqrt{\text{dim}(\alpha)}$. In the case that $i: N \hookrightarrow M$ is an inclusion of von Neumann algebras, the Q-system $A_{i} = (\overline{i} \circ i, \overline{i}(\overline{w}), w)$ is called the Q-system induced by $i$. In fact, every standard Q-system arises in this way:

\begin{theorem}[Reconstruction I] \label{recon1}
	Let $N$ be a properly infinite von Neumann factor and $A = (\theta,x,w)$ a $Q$-system in $\text{End}(N)$. Then, there exists a von Neumann algebra $M_A \equiv N \times_{\theta} A$ and an inclusion $i_A: N \hookrightarrow M_A$ such that $A$ is the $Q$-system induced by $i_A$. Moreover, we obtain a pair of conditional expectations
	\beq
		E_A: M_A \rightarrow N, \qquad E_A(m) = d_A^{-1} w^* \overline{i}_A(m) w,
	\eeq
	and 
	\beq
		\overline{E}_A: N \rightarrow N_1, \qquad \overline{E}_A(n) = d_A^{-1} x^* \overline{i}_A \circ i_A(n) x.
	\eeq
	Here, $N_1 \equiv \overline{E}_A(N)$ is called the canonical restriction of $N$. The latter conditional expectation is anti-isomorphic to a conditional expectation from $N'$ to $M_A'$ in a standard representation using the modular conjugation. 
\end{theorem}

Theorem \ref{recon1} establishes that a standard Q-system for the factor $N$ is in one to one correspondence with an inclusion $i_A: N \hookrightarrow M_A$ along with a pair of Kosaki dual conditional expectations $E_A$ and $\overline{E}_A$. This observation was originally made by Longo who realized that 
\begin{enumerate}
	\item The set of conditional expectations between a von Neumann algebra $M$ and a properly infinite subfactor $N$ are in one to one correspondence with the set of intertwining isometries between the identity endomorphism of $N$ and the canonical endomorphism $\gamma: M \rightarrow M$ restricted to $N$,
	\item For the same $N$ and $M$, and given a representation of $M$, the set of conditional expectations between $N'$ and $M'$ are in one to one correspondence with the set of intertwining isometries between the identity endomorphism of $M$ and the canonical endomorphism $\gamma: M \rightarrow M$. 
\end{enumerate}
The standard Q-system is born out of the interplay between a conjugate pair of isometric intertwiners $(w,\overline{w})$ for an inclusion $i: N \hookrightarrow M$ and its conjugate homomorphism $\overline{i}: M \rightarrow N$. In this case, the map $\gamma \equiv i \circ \overline{i}: M \rightarrow M$ is precisely the canonical endomorphism of the inclusion, and $\gamma\rvert_N \equiv \overline{i} \circ i: N \rightarrow N$ is the dual canonical endomorphism, which is anti-isomorphic (by modular conjugation) to the canonical endomorphism of the commutant inclusion $M' \subset N'$. Thus, $w \in \text{Hom}^0_{\text{End}(N)}(id_N, \gamma\rvert_N)$ and $\overline{w} \in \text{Hom}^0_{\text{End}(M)}(id_M,\gamma)$, which, by the above, are in one to one correspondence with conditional expectations $E_w: M \rightarrow N$ and $\overline{E}_{\overline{w}}: N' \rightarrow M'$.  

We would now like to confront the problem of constructing the algebra $M_A$ directly from the ingredients of the Q-system. $M_A$ is obtained by adjoining to $N$ a new element $v$ subject to the algebraic constraints:
\beq
	v i_A(n) = i_A \circ \theta(n) v, \; v^2 = i_A(x) v, \; v^* = i_A(w^*x^*) v, \; \forall n \in N. 
\eeq
We demand that any element in the algebra $M_A$ can be written in the form
\beq
	m = i_A(n) v, \; n \in N. 
\eeq
In this regard, we recognize $v$ as defining a Pimnser-Popa basis. Since we are working abstractly, it is not immediate that this union generates a $*$-algebra, let alone a von Neumann algebra. From this perspective, it is clear that one needs to impose conditions on $v$. In fact, the properties defining the Q-system are precisely the right conditions which need to be imposed to ensure $M_A$ has the appropriate algebraic structure.

The first thing we need is to impose some commutation relation between $v$ and included elements of $N$. This specifies the product structure on $M_{A}$. We take the simple relation
\begin{equation}
v i_A(n) = i_A \theta(n) v.
\end{equation}
In other words, passing $v$ through an included element just restricts that same element to the sector described by $\theta$. To define the product of $v$ with itself, we make use of the element $x$ in the Q-system to write:
\begin{equation}
    v^{2} = i_A(x) v. 
\end{equation}
If we want $M_A$ to be a von Neumann algebra, it must also be unital. To construct the identity element, we make use of the first Q-system property, $w^* x = \mathbb{1}_N$. Then, we find that for all $m = i_A(n) v \in M_A$:
\begin{align}
    i_A(w^{*})v i_A(n) v &= i_A(w^{*}) i_A \theta(n) v^{2}, \\
    &= i_A(w^{*} \theta(n)x) v , \\
    &= i_A(n w^{*}x) v , \\
    &= i_A(n) v , \implies i_A(w^{*})v = \mathbb{1}_{M_A}.
\end{align}

The algebra $M_{A}$ must also be associative. On the one hand we have
	\begin{align}
		(i_A(n_1 ) v i_A(n_2)) v i_A(n_3 ) v &= i_A(n_1 \theta(n_2) x \theta(n_3)x)v,
	\end{align}
	while on the other
	\begin{align}
		i_A(n_1 ) v (i_A(n_2) v i_A(n_3 ) v) &= i_A(n_1) v i_A(n_2 \theta(n_3) x) v, \\
		&= i_A(n_1 \theta(n_2) \theta^2(n_3) \theta(x)x) v.
	\end{align}
	Using the second condition for the Q-system, $\theta(x) x = x^2$, in the above equation gives
	\begin{align}
		i_A(n_1 ) v (i_A(n_2) v i_A(n_3 ) v) 
		&= i_A(n_1 \theta(n_2) \theta^2(n_3) x x) v \\
		&= i_A(n_1 \theta(n_2) x\theta(n_3)  x) v,
	\end{align} 
	and thus associativity is ensured.
 
At this point we have shown that $M_A$ has the structure of a unital, associative algebra. A von Neumann algebra is moreover involutive, so we have to construct an involution. The involution of elements $i_A(n) \in M_A$ is induced from the involution on $N$, so it remains to define an involution for $v$. We take
\begin{equation}
    v^{*}=  i_A(w^{*}x^{*})v. 
\end{equation}
For an arbitrary element $m = i_A(n) v \in M_{A}$, we must ensure that its involution $m^{*}$ is also in $M_{A}$. This is immediate, as:
\begin{align}
    \left( i_A(n) v\right)^{*} &= v^{*} i_A(n)^{*}, \nonumber \\
    &= i_A(w^{*}x^{*})v i_A(n)^{*},  \nonumber\\
    &= i_A(w^{*}x^{*} \theta(n)^{*})v \nonumber.
\end{align}
Moreover, we need to ensure that the involution is antimultiplicative. We first simplify the product of two elements and then take the adjoint
\begin{align}
     \left[(i_A(n_{1}) v) (i_A(n_{2})v)\right]^{*} &= \left[ i_A(n_{1} \theta(n_{2}) x) v \right]^{*}, \nonumber \\
     &= i_A\left[ w^{*} x^{*} \theta(x^{*} \theta(n_{2}^{*}) n_{1}^{*}) \right] v, \nonumber \\
     &= i_A\left[ \theta(n_{1}\theta(n_{2}) x) x w \right]^{*} v,  \nonumber\\
     &= i_A\left[ \theta(n_{1}) \theta^2(n_{2}) x x w \right]^{*} v, \nonumber \\
     &= i_A\left[ \theta(n_{1}) x\theta(n_{2}) x  w \right]^{*} v, \nonumber \\
     &= i_A\left[  w^{*}x^{*} \theta(n_{2}^* ) x^{*} \theta(n_{1}^* ) \right] v, \nonumber
\end{align}
where we used the second property $\theta(x) x = x^2$ and the fact that $x$ intertwines $\theta$ with its square. To verify that it is anti-multiplicative, we compute
\begin{align}
    \left[(i_A(n_{1}) v) (i_A(n_{2})v)\right]^{*} &=  (i_A(n_{2})v)^{*} (i_A(n_{1})v)^{*},  \nonumber\\
   &=i_A\left[ w^{*}x^{*} \theta(n_{2}^{*} w^{*}x^{*} \theta(n_{1}^{*})) x\right] v, \nonumber \\
   &=i_A\left[x^{*} \theta( \theta(n_{1} x w n_{2}) x w \right]^{*} v,  \nonumber\\
   &=i_A\left[\theta(n_{1}) x x^{*} \theta( w n_{2}) x w \right]^{*} v, \nonumber \\
   &=i_A\left[ w^{*}x^{*} \theta(n_{2}^{*}) \theta(w^{*}) x x^{*} \theta(n_{1}^{*}) \right] v, \nonumber \\
   &=i_A\left[ w^{*}x^{*} \theta(n_{2}^{*}) x^{*} \theta(n_{1}^{*}) \right] v, \nonumber
\end{align}
where we used the fact that $x^{*}$ intertwines $\theta^{2}$ with $\theta$ and the final property of the Q-system, $x^{*} \theta(x) = x x^{*}$. 

In summary, we see that each of the conditions defining the Q-system translate to an algebraic property of the extended algebra $M_{A}$. The first property ensures the existence of a unit, the second property ensures the associativity of the product, and the third property allows us to define the involution. Finally, the existence of a conditional expectation $E_A: M_A \rightarrow N$ endows upon $M_A$ the proper topology required by a von Neumann algebra. Succinctly, $E_A$ acts as a state extension, which allows for standard representations of $M_A$ to be induced (in the sense of Rieffel) from standard representations of $N$.

\section{Q-Systems for Operator Valued Weights} \label{sec: QOVW}

In the previous sections we have explored the rich structure endowed upon von Neumann algebra when it admits a conditional expectation. Unfortunately, there are many instances in the study of von Neumann algebras in which conditional expectations are not accessible. These scenarios coincide with inclusions of infinite indexes. In physical applications like quantum field theory and holography, infinite index inclusions appear to be the norm rather than an exception. This is ample reason to try to port over the tools we have developed to the infinite index case.

A paradigmatic example of such an inclusion is the crossed product of a von Neumann algebra $N$ by a locally compact group $G$. This construction has emerged as fundamental in the study of quantum gravity, where it has been argued that the algebra of a gravitational subregion can often be understood as the modular crossed product of QFT subregion algebra with its modular automorphism group. As we have addressed in Section \ref{sec: CP primer}, the inclusion associated with a crossed product only admits a conditional expectation if the group in question is finite. This is intimately related to the fact that the index of this would-be conditional expectation is equal to the order of the group. Given an infinite group we are left with an operator valued weight. Nevertheless, we see that, irrespective of the group $G$, the structure of the crossed product bears a very close resemblance to the Q-system analysis. For example, in eqn. \eqref{CP Basis} the group generators $\lambda(g)$ continue to play the role of a Pimsner-Popa basis.

This is not a coincidence. In fact, all of the features of the Q-system introduced in Section \ref{sec: QCE} have analogs in the case of an inclusion admitting only an operator valued weight, provided we take sufficient care in our analysis. In this section, we undertake such an analysis. We formally define our notion of a Spatial Q-system which provides an alternative classification of inclusions, $N \subset M$, admitting operator valued weights in terms of the existence of a compatible triple consisting of a representation, a Pimsner-Popa basis, and an intertwiner. We discuss how this data can be regarded as intrinsic to the algebra $N$, and prove a reconstruction theorem which allows us to construct the extension algebra $M$ purely from this data. 

\subsection{Some subtleties}

To arrive at our definition of the Spatial Q-system, we will follow essentially the same procedure as we did in Section \ref{sec: QCE}, but working in the case where $E$ is replaced by an operator valued weight $T$. As we will see, this has profound implications for the structure of the ingredients appearing in the Q-system. 

Let $i: N \hookrightarrow M$ be an inclusion of von Neumann algebras admitting an operator valued weight $T \in P(M,N)$. Recall that the existence of $T$ implies the existence of faithful, semifinte, normal weights $\varphi \in P(M)$ and $\varphi_0 \in P(N)$ such that $\varphi = \varphi_0 \circ T$. Let $L^2(M;\varphi)$ and $L^2(N;\varphi_0)$ be the GNS Hilbert spaces of $M$ and $N$ with respect to these weights. Already, an important distinction arises as compared to the previous case; the Hilbert space $L^2(N;\varphi_0)$ cannot, in a rigorous sense, be regarded as a subspace of $L^2(M;\varphi)$. We will develop this point more closely, but the culprit is the unboundedness of the operator valued weight. In fact, this will be a central ingredient in the subsequent analysis. 

Despite the complication alluded to above, we can still use the operator valued weight $T$ to define an $N$-valued bilinear on\footnote{Again, by an abuse of notation, we will often regard $G_T$ as a map on $M^{\times 2}$.} $\mathfrak{n}_T$
\beq
	G_T: \mathfrak{n}_T^{\times 2} \rightarrow N, \qquad G_T(m_1,m_2) = T(m_1^* m_2). 
\eeq
The same Gram-Schmidt procedure that applied in the conditional expectation case continues to apply here, and we can once again construct a Pimsner-Popa basis\footnote{In principle, the index set $\mathcal{I}$ may be uncountable.} $\{\lambda_i\}_{i \in \mathcal{I}} \subset M$ which realizes $M$ as a von Neumann union $M \simeq i(N) \vee \{\lambda_i\}_{i \in \mathcal{I}}$. This was first observed by Herman and Ocneanu \cite{herman1989index}. 

Where the subtlety does arise, however, is in formulating the intertwiner $W$ that implements the operator valued weight. In the conditional expectation case, $W$ was intimately related to the Jones projection. That $L^2(N;\varphi_0)$ is not a strict Hilbert subspace of $L^2(M;\varphi)$ manifests itself in the fact that there is not a well defined Jones projection affiliated wtih an operator valued weight. Not unrelated, the intertwiner $W$ ceases to be an isometry. Instead, $W$ must be regarded as a densely defined, unbounded map. To see that this is the case, we can follow the same procedure as in Section \ref{sec: QCE}. 

We define
\beq \label{OVW Intertwiner}
	W: L^2(N;\varphi_0) \rightarrow L^2(M;\varphi), \qquad W(\eta_{\varphi_0}(n)) = \eta_{\varphi} \circ i(n). 
\eeq
For most $n \in N$, the operator $i(n) \in M$ does not belong to the domain of $\varphi$. Thus, the notation $\eta_{\varphi} \circ i(n)$ is purely formal, since $i(n)$ is not in the domain of $\eta_{\varphi}$. Nevertheless, we can still make sense of what Eq.~\eqref{OVW Intertwiner} is saying -- to each $n \in N$ we associate a possible non-normalizable vector $\eta_{\varphi}(n)$ associated with $L^2(M;\varphi)$.\footnote{To make this rigorous, we could make use of the rigged Hilbert space associated with $L^2(M;\varphi)$.} By construction, Eq.~\eqref{OVW Intertwiner} is an unbounded intertwiner $W \in \text{Hom}_N(\pi_{\varphi} \circ i, \pi_{\varphi_0})$. Here, again, $\pi_{\varphi} \circ i$ should be read as a representation of $N$ acting as possible unbounded operators affiliated with $B(L^2(M;\varphi))$. 

We can still compute the formal adjoint of $W$, and since $T$ possesses the same homogeneity property as a conditional expectation the computation between Eq.~\eqref{Adjoint intertwiner} and Eq.~\eqref{Adjoint of W} holds. That is, $W^{\dagger} \in \text{Hom}_N(\pi_{\varphi_0}, \pi_{\varphi} \circ i)$ acts as
\beq
	W^{\dagger}(\eta_{\varphi}(m)) = \eta_{\varphi_0} \circ T(m). 
\eeq 
In this respect, the intertwiner $W$ implements the operator valued weight as
\beq
	T(m) = \pi_{\varphi_0} \circ \text{Ad}_{W^{\dagger}} \circ \pi_{\varphi}(m),
\eeq
which is well defined for $m \in \mathfrak{n}_{T}$. However, the operator valued weight is not a unital map and thus $W$ ceases to be an isometry:
\beq
	\mathbb{1} \neq W^{\dagger} \pi_{\varphi}(\mathbb{1}) W = W^{\dagger} W. 
\eeq

\subsection{The Spatial Q-System} \label{sec: spatial recon}

In summary, the existence of an operator valued weight coincides with the emergence of a Pimsner-Popa basis along with an unbounded and non-isometric intertwiner. In Section \ref{sec: comp}, we will discuss the non-isometricity of the map $W$ in the context of quantum error correction. For now, we would like to organize our findings in a manner which makes reference only to the algebra $N$. In this way, we can regard this construction as defining an extension algebra which admits an operator valued weight. This leads us to the following definition:

\begin{definition}[Spatial Q-System]
	Let $N$ be a von Neumann algebra and $\varphi_0 \in P(N)$ a faithful, semifinite, normal weight with vector representative $\xi_{\varphi_0} \in L^2(N;\varphi_0)$. A \textbf{Spatial Q-system} for $N$ is a triple $(\pi, \{\Lambda_i\}_{i \in \mathcal{I}}, W)$ consisting of
	\begin{enumerate}
		\item A `representation' $\pi: N \rightarrow \widehat{B(H)}$, where here $\widehat{B(H)}$ is the set of possibly unbounded operators on the Hilbert space $H$,
		\item An intertwiner $W \in \text{Hom}_N(\pi_{\varphi_0},\pi)$, and
		\item A family of operators $\{\Lambda_i\}_{i \in \mathcal{I}} \subset B(H)$ satisfying
		\begin{enumerate}
			\item $W^{\dagger} \Lambda_i \Lambda_j^{\dagger} W \in \pi_{\varphi_0}(N)$ for all $i,j \in \mathcal{I}$, and  
			\item $H = \sum_{i \in \mathcal{I}} \Lambda_i^{\dagger} \pi(n) W(\xi_{\varphi_0})$.
		\end{enumerate}
	\end{enumerate}
	The two properties in condition $(3)$ are, in fact, equivalent to the those defining a Pimsner-Popa basis \cite{fidaleo_canonical_1999}.
\end{definition}

The role of a Q-system is to induce extensions of $N$, as we saw explicitly in Section \ref{sec: Longo}. The spatial Q-system also has an extension theory encapsulated in the following theorem:
\begin{theorem}[Reconstruction II]
    Let $N$ be a von Neumann algebra and $A = (\pi, \{\Lambda_{i}\}, W)$ a Spatial Q-system. Then, there exists an algebra $M_{A}$ extending $N$ along with an operator valued weight $T_{A} \in P(M_A, N)$. 
\end{theorem}
\begin{proof}
    We define the algebra $M_{A} \equiv N \times_{\pi} A$ to be the von Neumann union of $\pi(N)$ with the Pimsner-Popa basis, namely
    \begin{equation}
        M_{A} := \pi(N) \vee \{\Lambda_{i}^{\dagger}\},
    \end{equation}
    which is the natural generalization of the construction in the standard Q-system. A generic element $m \in M_A$ can be decomposed as
    \beq
    		m = \sum_{i \in \mathcal{I}} \Lambda_i^{\dagger} \pi(n_i),
    \eeq
    where $\{n_i\}_{i \in \mathcal{I}} \subset N$. 

    It remains to show that $A$ induces an operator-valued weight. Consider the map
    \beq
    		\mathcal{T}_A: M_A^+ \rightarrow \widehat{B(L^2(N;\varphi_0))}_+, \qquad \mathcal{T}_A(m) \equiv W^{\dagger} m W. 
    \eeq
    Clearly, $\mathcal{T}_A$ is a linear and positive map. For it to be an operator valued weight its range has to be $\widehat{\pi_{\varphi_0}(N)}$, and it must be homogeneous when acting upon elements in $\pi(N)$. The homogeneity of the map $\mathcal{T}_A$ is ensured by the intertwining property of $W$. Thus, it remains only to show that it has the appropriate range. A generic positive element in $M_A$ is given by 
    \beq
    		m^{\dagger} m = \sum_{i,j \in \mathcal{I}^{\times 2}} \pi(n_i^*) \Lambda_i \Lambda_j^{\dagger} \pi(n_j).
    \eeq
    Acting upon such an operator with $\mathcal{T}_A$ we find
    \begin{align}
        W^{\dagger} m^{\dagger} m  W = \sum_{(i,j) \in \mathcal{I}^{\times 2}} \pi_{\varphi}(n_i^*) W^{\dagger} \Lambda_i \Lambda_j^{\dagger} W \pi_{\varphi_0}(n_j).
    \end{align}
    Here, we have utilized the intertwining property of $W$. The fact that $\Lambda_i$ form a Pimsner-Popa basis implies that $W^{\dagger} \Lambda_i \Lambda_j^{\dagger} W$ is in the image of $\pi_{\varphi_0}$ and thus $\mathcal{T}_A(M_A) \subset \pi_{\varphi_0}(N)$. We can therefore promote $\mathcal{T}_A$ to an operator valued weight
    \beq
    		T_A \equiv \pi_{\varphi_0}^{-1} \circ \mathcal{T}_A,
    \eeq
    which concludes the proof.
\end{proof}

\section{Comparison of Q-systems} \label{sec: comp} 

We would now like to compare and contrast the three notions of Q-system which we have introduced in this text. There are several dimensions along which these constructions can be compared -- for our purposes we would like to provide one mathematically oriented comparison and two physically oriented comparisons. Mathematically, it is interesting to regard the different kinds of Q-systems according to their categorical structure. As we have remarked, the standard Q-system of Longo possesses a very deep categorification. On the other hand, the more abstract Q-systems are capable of describing more general scenarios at the expense of the categorical structure. Nonetheless, it seems possible that the Spatial Q-system may possess a categorical structure if we are willing to work in a more flexible category than was regarded for the standard case. Physically, it is interesting to compare the different kinds of Q-systems in terms of their information theoretic capacities to perform quantum error correction or in terms of their role as probes or quantum reference frames of a system.

\subsection{Categorification of Q-systems}

A standard Q-system for the algebra $N$ is a Frobenius algebra in the category of endomorphisms $\text{End}(N)$. Really, the appropriate category is $\text{End}_0(N)$ -- which is the category of finite index endomorphisms. By a similar token, although the generalized and spatial Q-systems lack the fine categorical structure as the standard case, they both center around relaxations of the category $\text{End}_0(N)$. In particular, the central object in the generalized Q-system is a possibly infinite index endomorphism $\theta \in \text{End}(N)$, and the central object in the spatial Q-system is a representation $\pi \in \text{Mod}(N)$, where here $\text{Mod}(N)$ is the category of representations of $N$. It is not difficult to show that, as categories, 
\beq \label{Categorical Nesting}
	\text{End}_0(N) \subset \text{End}(N) \subset \text{Mod}(N).
\eeq

The crucial difference between the category of endomorphisms and the category of representations is the structure of the arrows. In the endomorphism category, arrows between endomorphisms were taken to be elements of the algebra. Conversely, in the module category arrows between representations are operators between Hilbert spaces. As we have observed, this difference is crucial in allowing for the transition from conditional expectations to operator valued weights since in the latter case the GNS Hilbert spaces of the algebras involved are not related by an orthogonal projection. In this regard, we cannot avoid the need to formulate our analysis with respect to at least two distinct Hilbert spaces. 

Having made these observations, a natural question emerges as to whether one can understand the Spatial Q-system as a formal structure within the category $\text{Mod}(N)$. Some hint at how this may be done can be seen in the analysis of \cite{Bischoff:2014xea}, in which the same question is asked of the generalized Q-system. The authors proposed a refinement of the generalized Q-system $(\theta,\{\nu\}_{i \in \mathcal{I}}, w)$ called a generalized Q-system of intertwiners where the original structure is enhanced by requiring that each $\nu_i$ is itself an intertwiner in $\text{Hom}_{\text{End}(N)}(\theta,\theta^2)$. This mimics the quality of the Pimsner-Popa basis in the standard Q-system, and gives the generalized Q-system the complexion of a collection of Frobenius algebras inside $\text{End}(N)$ along with some fusion rules for multiplying between them. This property can then be used to define a commutator between the Pimsner-Popa basis and elements of the subalgebra $N$. This is in contrast to the Spatial Q-system in which the Pimsner-Popa basis consists of linear operators on a Hilbert space within which the algebra $N$ is also represented, thereby automatically giving rise to a product structure. It is only natural to then wonder what the categorical avatar of a Pimsner-Popa basis is in a general category. We leave the question of investigating the full categorical structure of Spatial Q-systems within $\text{Mod}(N)$ to future work. 

\subsection{Q-systems as Quantum Error Correcting Protocols}

In a very similar vein to the above discussion, the notions of Q-system which we have introduced also exhibit a nested quantum information theoretic structure. The standard Q-system is a device for encoding dual conditional expectations mapping from the extension algebra $M$ to the subalgebra $N$, and mapping from the commutant $N'$ to the commutant $M'$. In this regard, a standard Q-system exhibits the structure of exact complementary error correction as discussed, for example, in \cite{Faulkner:2020hzi}. In words, the standard Q-system provides a protocal for recovering the algebra $N$ as it is embedded within $M$ \emph{and} for the recovering the commutant algebra $M'$ as it is embedded within $N'$. When we move to the generalized Q-system, we lose one of the two conditional expectations. Thus, the generalized Q-system describes exact quantum error correction, but not of the complementary variety. We can either exactly recovery $N \subset M$ or $M' \subset N'$. Finally, the Spatial Q-system encodes operator valued weights. As we have discussed, operator valued weights (when they are not conditional expectations) are non-unital maps. At the Hilbert space level, this implies that the operator valued weight is implemented by a non-isometric map. As has been explored in recent work \cite{akers2022blackholeinteriornonisometric}, non-isometric maps are a device for encoding approximate error correction. 

If we denote by $C_0(M,N)$ the set of conditional expectations of finite index, the above discussion implies that the three Q-systems we have introduced also fit into a nested sequence:
\beq \label{QI nested sequence}
	C_0(M,N) \subset C(M,N) \subset P(M,N). 
\eeq
Indeed, the constructions we have laid out can largely be thought of as defining a bijection between Eqs.~\eqref{Categorical Nesting} and \eqref{QI nested sequence}. In this regard, the Spatial Q-system seems to present a very powerful framework for organizing our thinking about quantum error correction. Moreover, it shows us that, even in the absence of exact quantum error correction, there is still a very intimate relationship between algebras in the case of extensions and inclusions. 
\subsection{Quantum reference frames}
Yet another physical facet of the Q-system can be understood through quantum reference frames. There is a growing literature on the relation between algebraic quantum field theory and quantum reference frames~\cite{Fewster:2024pur,DeVuyst:2024pop,AliAhmad:2024wja }. One way of describing these frames is to start with a  system with a time translation symmetry and adjoining a clock that is covariant with respect to the group that conditions the system on its recorded time~\cite{pwootters}. The combined theory with the system and the clock has as its physical algebra of observables the crossed product of the system with the time translations and so is described by a Q-system. The Q-system may then be associated with a generalized quantum reference frame, in which the extension is not necessarily governed by a symemtry object like a group, for a given system. It is a way of probing aspects of the system of interest in a consistent manner with its structure. It is interesting to consider then quantum reference frames with respect to generalized symmetries that are not described by a group, which will be given by specific spatial Q-systems in the general setting as the symmetry object may not necessarily be finite. 
\section{Discussion} \label{sec: disc}

In this note we have proposed a far-reaching extension of Longo's Q-system which provides a spatial characterization of inclusions of von Neumann algebras admitting faithful, semifinite, normal operator valued weights. Given a von Neumann algebra $N$, our Spatial Q-system provides an indexing of the possible extensions $N \subset M$ without any assumption of finiteness of the inclusion or the existence of conditional expectations. In this regard, the extension of the algebra $N$ which is defined by specifying a Spatial Q-system may be interpreted as a vast generalization of the notion of the crossed product which extends the algebra $N$ by a collection of operators comprising a Pimsner-Popa basis but need not have the structure of a locally compact group. 

In this discussion we would like to exhibit several physical scenarios in which the toolset provided by the Spatial Q-system may be of great use. In future work we intend to investigate these directions. 
\begin{enumerate}
   \item \textbf{Extensions by exotic symmetry objects:}  We have paid a lot of attention in this note to extensions which arise from locally compact groups, but there are many physical systems which admit symmetries of more exotic varieties. For example, condensed matter systems have been known to exhibit categorical and non-invertible symmetries \cite{McGreevy_2023}. The spatial Q-system is a natural candidate for classifying extension by these more general symmetry objects. The Doplicher-Roberts reconstruction theorem  reconstructs a kinematical theory with a compact group from the neutral part and superselection structure data. In the language of this note, this corresponds to a finite Q-system extending the neutral algebra by charged intertwiners compatible with the group. One could use this philosophy to obtain extension of theories by something more general than a group, namely an algebra or a category. In a forthcoming note, we prove a reconstruction theorem for fusion categorical symmetries in the spirit of Doplicher-Roberts and provide the algebraic grounds for a description of these symmetries in terms of canonical methods. 
   \item \textbf{Entropic order parameters in infinite dimensions:} There is a fruitful framework for computing entropic order parameters alerting us of a symmetry breaking pattern in quantum field theory~\cite{Casini:2020rgj}. In particular, one such parameter is the entropy of a state relative to its invariant part under the action of a finite group. Using the canonical description of a theory with its extension by our \textit{spatial} Q-system, one can now describe order parameters for the symmetry breaking of infinite dimensional objects like locally compact groups and infinite non-invertible symmetries. In a forthcoming note, we describe this framework in more detail and work out explicit physical examples.
    \item \textbf{Parent theories:} The original Q-system program was fruitful in 2D CFT for classifying \textit{new} theories in which a given theory could be included while maintaining its locality and causality. As an example, the Ising model is known to have two Q-system which give rise to two parent theories~\cite{Bischoff_2015}. The spatial Q-system opens up the possibility of even more general parent theories that respect the underlying structure of a given one. These parent theories will no longer be bound by finiteness conditions and so it will be interesting to investigate them to better understand the vast landscape of quantum field theory. 
       \item \textbf{Quantum error correction:} It has already been remarked in the bulk of the note that the spatial Q-system has a quantum error correcting interpretation. Indeed, it describes the encoding of the original algebra, which we can now call the logical algebra, into a larger physical algebra. The operator-valued weight one gets from the spatial Q-system is essentially a recovery map. However, if one works strictly with such an object, then the correction \textit{cannot} be isometric or exact. The three cases outlined in the introduction take on the following interpretation: (1) exact complementary error correction, (2) exact error correction that is not complementary, and (3) approximate error correction. The spatial Q-system provides a natural language to analyze approximate, non-isometric quantum error correction, which plays a vital role, for example, in theories of the black hole interior \cite{akers2022blackholeinteriornonisometric}. There is also an interesting connection between complexity and non-locality of the operators involved in the  extension by spatial Q-system which would be relevant for these considerations.
    \item \textbf{Generalized entropy:} The crossed product has provided a way to understand generalized entropy in a Lorentzian and algebraic manner. In our language, the generalized entropy comes from a contribution from the original algebra $N$ in addition to one associated with the information contained in the Pimsner-Popa basis induced by an appropriate operator valued weight. In the case of a finite conditional expectation, this contribution is a central element which has been identified with the area in holographic settings~\cite{Harlow:2016vwg,faulkner2020holographicmapconditionalexpectation}. However, on physical grounds, the area is \textit{not} a central operator due to backreaction and non-perturbative effects in quantum gravity. In other words, $N$ and the area operator influence each other since an excitation from the matter sector can lead to a change in the region of its localization. In previous work related to this note~\cite{AliAhmad:2024saq}, we quantify this effect. It is manifest that the `area operator' is not central and there are additional terms which we interpret as backreaction. This provides an algebraic, Lorentzian, and physical framework for the emergence of the quantum extremal surface proposal of holography but also corrections to the generalized entropy in full quantum gravity. We also expect that this framework should be capable of quantifying more general contributions to the entropy like topological terms. 
    \item \textbf{Holography:} Holography may also be understood as an inclusion, in this case of a particular bulk algebra inside an appropriate subregion algebra of its CFT dual. The general goal of holography is to work out the dictionary translating between the bulk and boundary theories. However, a necessary prerequisite is determining the appropriate bulk algebra which is holographically dual to the CFT subregion. This raises tricky questions about the locality of operators in these algebras. In the bulk, one can have local inclusions of observables if there is an inclusion of subregions. However, these local inclusions (as long as they are bounded by the quantum extremal surface) cannot translate to local inclusions on the boundary since they are defined on the same subregion. A version of this was explained in Ref.~\cite{Bao:2024hwy}. This points to extensions of the boundary algebra that correspond to observables which reach deeper into the bulk. This may give us a handle for asking questions regarding the interior of black holes by seeing how far one can extend the boundary algebra. At the same time, the spatial Q-system provides a tool for indexing the space of extensions of a given algebra and so may provide some structure to the problem of systematically constructing holographically dual algebras. This suggests a route towards a purely boundary justification for entanglement wedge reconstruction. Even more optimistically, understanding the `right' reconstruction scheme for AdS/CFT may allow us to port what we learn towards a notion of holography for general spacetimes~\cite{Bousso:2022hlz,bousso_holograms_2023}.
    \item \textbf{Algebraic quantum gravity:} There is still the question of the relevance of algebras in the description of a non-perturbative theory of quantum gravity~\cite{Stewart:1974uz, Fredenhagen:1989kr, Haag:1990ht, Brunetti:2013maa, Rejzner:2016yuy, Brunetti:2016hgw, Brunetti:2022itx, AliAhmad:2025ukh}. If one replaces the bulk and boundary in the above paragraph, or the logical and physical in the quantum error correction one, with the semiclassical and quantum gravitational, we are led to positing an inclusion of semiclassical algebras inside of quantum gravitational ones. This is also resonant with the expectation from the path integral in the following sense: the path integral over matter and gravity contains within it the path integral over matter about a fixed (or perturbatively fluctuating) gravitational background. One can then further constrain the inclusion structure and the data specifying the extension into the quantum gravitational algebra, namely the spatial Q-system, based on requirements on general grounds of quantum gravity. For example, these can be swampland constraints translated to the algebraic language like the absence of global symmetries in quantum gravity. In future work, we explore the general setup described here  and apply it to the problem of black hole evaporation. 
\end{enumerate}

\section*{Acknowledgments}
The authors would like to Thomas Faulkner, Rob Leigh, Roberto Longo, and Yifan Wang for interesting and insightful discussions.


\appendix

\section{Preliminaries} \label{app: prelim} 
In this section, we provide background material underpinning our investigations.

\subsection{Spatial/Modular Theory}

Given a faithful, semifinite, normal state $\varphi \in C(M)$ we can construct the GNS representation $L^2(M;\varphi)$. Formally, $L^2(M;\varphi)$ is the completion of the vector space\footnote{If we did not assume that $\varphi$ was faithful, we would need to quotient by its kernel.} $\mathfrak{n}_{\varphi}$ with respect to the inner product
\beq
	g_{L^2(M;\varphi)}(\eta_{\varphi}(m_1), \eta_{\varphi}(m_2)) \equiv \varphi(m_1^* m_2), \qquad m_1,m_2 \in \mathfrak{n}_{\varphi}. 
\eeq
Here, the notation $\eta_{\varphi}(m)$ simply serves to distinguish $m$ as an element of the Hilbert space. Since $\mathfrak{n}_{\varphi}$ is a two sided ideal in $M$ we automatically get a representation $\pi_{\varphi}: M \rightarrow B(L^2(M;\varphi))$ and an anti-representation $\overline{\pi}_{\varphi}: M \rightarrow B(L^2(M;\varphi))$ given by
\beq
	\pi_{\varphi}(m_1) \eta_{\varphi}(m_2) \equiv \eta_{\varphi}(m_1 m_2), \qquad \overline{\pi}_{\varphi}(m_1) \eta_{\varphi}(m_2) \equiv \eta_{\varphi}(m_2 m_1). 
\eeq
We will often make use of the GNS representation. 

A fundamental tenet of the study of von Neumann algebras is modular theory \cite{takesaki2006tomita}. Here, we shall discuss Connes' spatial approach \cite{connes1980spatial,ENOCK1996466}. Let $\pi: M \rightarrow B(H)$ be a representation of $M$ and $\varphi \in P(M)$. The lineal of $\varphi$, $D(H;\varphi)$, is defined to be the set of elements $\xi \in H$ which are `bounded' with respect to $\varphi$ in the sense that
\beq	 \label{Lineal condition}
	\exists C_{\xi} \in \mathbb{R}_+ \ \{\infty\} \; | \; \norm{\pi(m) \xi}_{H} \leq C_{\xi} \varphi(m^* m), \qquad \forall m \in \mathfrak{n}_{\varphi}. 
\eeq
To each $\xi \in D(H;\varphi)$ we can associate a linear map $K^{\varphi}_{\xi}: L^2(M;\varphi) \rightarrow H$ defined by
\beq
	K^{\varphi}_{\xi}( \eta_{\varphi}(m) ) \equiv \pi(m) \xi. 
\eeq	
By Eq.~\eqref{Lineal condition}, the map $K^{\varphi}_{\xi}$ is bounded. Moreover, by construction it is an intertwiner between the representations $\pi_{\varphi}: M \rightarrow B(L^2(M;\varphi))$ and $\pi: M \rightarrow B(H)$ in the sense that
\beq \label{Connes Spatial Operator}
	K^{\varphi}_{\xi} \pi_{\varphi}(m) \eta_{\varphi}(m') = K^{\varphi}_{\xi} \eta_{\varphi}(m m') = \pi(m m') \xi = \pi(m) \pi(m') \xi = \pi(m) K^{\varphi}_{\xi} \eta_{\varphi}(m'), \qquad \forall m \in M. 
\eeq
It is worth noting that the map $K^{\varphi}_{\xi}$ also makes sense for vectors $\xi \in H$ which are \emph{not} bounded with respect to $\varphi$. In this instance, $K^{\varphi}_{\xi}$ becomes an unbounded operator with dense domain.

For future reference, given an algebra $M$ with a pair of representations $\pi_i: M \rightarrow B(H_i)$, we denote by $\text{Hom}_M(\pi_1,\pi_2)$ the set of intertwining maps $U: H_1 \rightarrow H_2$ such that
\beq
	U \pi_1(m) = \pi_2(m) U, \qquad \forall m \in M, \xi' \in H. 
\eeq	
We further denote by $\text{Hom}^0_M(\pi_1,\pi_2)$ the set of bounded intertwiners. Thus, we can write $K^{\varphi}_{\xi} \in \text{Hom}^0_M(\pi, \pi_{\varphi})$. Given $U \in \text{Hom}_M(\pi_1,\pi_2)$ and $V \in \text{Hom}_M(\pi_2,\pi_3)$ we have
\beq
	VU \pi_1(m) = V \pi_2(m) U = \pi_3(m) VU, \qquad \forall m \in M.,
\eeq
thus the product of intertwiners $VU \in \text{Hom}_M(\pi_1,\pi_3)$. If the intertwiners are both bounded, their product will be too.

Let $K^{\varphi}_{\xi}{}^{\dagger}: H \rightarrow B(L^2(M;\varphi))$ denote the formal adjoint of $K^{\varphi}_{\xi}$. By properties of the adjoint, $K^{\varphi}_{\xi}{}^{\dagger} \in \text{Hom}^0_M(\pi_{\varphi},\pi)$. Thus, for any pair of vectors $\xi_1, \xi_2 \in D(H;\varphi)$ the operator $\Theta_{\xi_1,\xi_2}^{\varphi} \equiv K^{\varphi}_{\xi_1} K^{\varphi}_{\xi_2}{}^{\dagger}: H \rightarrow H$ is bounded and belongs to the space $\text{Hom}^0_M(\pi,\pi)$. The set of bounded intertwiners from a representation to itself is equivalent to the commutant of $M$ in the given representation:
\beq
	\text{Hom}^0_M(\pi,\pi) = \{\mathcal{O} \in B(H) \; | \; \mathcal{O} \pi(m) = \pi(m) \mathcal{O}, \; \forall m \in M\} = \pi(M)'. 
\eeq	
Thus, $\Theta_{\xi_1,\xi_2}^{\varphi} \in \pi(M)'$. More suggestively, we can regard $\Theta^{\varphi}: B(H) \rightarrow \pi(M)'$ such that
\beq \label{Theta as an OVW}
	\xi_1 \otimes \xi_2^* \mapsto \Theta^{\varphi}_{\xi_1,\xi_2}.
\eeq	

Given a weight $\Psi \in P(\pi(M)')$ we can compute the expectation value $\Psi(\Theta_{\xi_1,\xi_2}^{\varphi})$. Using Eq.~\eqref{Theta as an OVW} we can interpret this as the expectation value of a projection operator in $B(H)$:
\beq		
	\Psi(\Theta_{\xi_1,\xi_2}^{\varphi}) = \Psi \circ \Theta( \xi_1 \otimes \xi_2^*).
\eeq
Since $B(H)$ is a type I algebra, we know that there exists a positive and self adjoint (although possible unbounded) operator $\frac{d\Psi}{d\varphi}: H \rightarrow H$ such that
\beq
	\Psi \circ \Theta^{\varphi}(\xi_1 \otimes \xi_2^*) = tr_H(\frac{d\Psi}{d\varphi} \xi_1 \otimes \xi_2^*) = g_H(\xi_1, \frac{d\Psi}{d\varphi} \xi_2). 
\eeq
This operator is called the \emph{spatial derivative} of $\Psi$ with respect to $\varphi$. We note that the spatial derivative as the property
\beq
	\frac{d\Psi}{d\varphi} = \frac{d\varphi}{d\Psi}^{-1}, \qquad \forall \Psi \in P(\pi(M)'), \; \varphi \in P(M).
\eeq

The spatial derivative can be used to define the modular automorphism and the Connes' cocycle in a representation that need not be standard. Given $\varphi \in P(M)$, the modular automorphism is a map $\sigma^{\varphi}: \mathbb{R} \rightarrow \text{Aut}(M)$ which is defined by
\beq \label{Modular Aut}
	\pi \circ \sigma^{\varphi}_t(m) \equiv \frac{d\varphi}{d\Psi}^{it} \pi(x) \frac{d\varphi}{d\Psi}^{-it},
\eeq
for any $\Psi \in P(\pi(M)')$. Given any pair of weights $\varphi_1,\varphi_2 \in P(M)$ the Connes' cocycle is a one parameter family of operators $u^{\varphi_1 \mid \varphi_2}_t \in M$ defined by
\beq \label{Connes' cocycle}
	\pi(u^{\varphi_1 \mid \varphi_2}_t) = \frac{d\varphi_1}{d\Psi}^{it} \frac{d\varphi_2}{d\Psi}^{-it},
\eeq
where again $\Psi$ can be chosen arbitrarily. Recall that the Connes' cocycle intertwines the modular automorphisms of the weights:
\beq
	\sigma^{\varphi_1}_t(m) = u^{\varphi_1 \mid \varphi_2}_t \sigma^{\varphi_2}_t(m) u^{\varphi_1 \mid \varphi_2}_{-t},
\eeq
which is the statement that the modular automorphisms of weights are equivalent up to conjugation by a unitary.

If we take $H = L^2(M;\varphi)$ to also be a GNS representation of $M$, then $\pi_{\varphi}(M)' \simeq M$ and any weight $\psi \in P(M)$ gives rise to a weight on the commutant $\psi' \in P(\pi_{\varphi}(M)')$ as
\beq
	\psi'(\mathcal{O}) \equiv g_{L^2(M;\varphi)}(\xi_{\psi}, \mathcal{O} \xi_{\psi}), \qquad \forall \mathcal{O} \in \pi_{\varphi}(M)',
\eeq
where here $\xi_{\psi} \in L^2(M;\varphi)$ is called the vector representative of $\psi$ with respect to $\varphi$. In this case the spatial derivative reproduces the relative modular operator introduced by Araki:
\beq
	\frac{d\psi'}{d\varphi} = \Delta_{\psi \mid \varphi}. 
\eeq
Then the Eqs.~\eqref{Modular Aut} and \eqref{Connes' cocycle} revert to their more familiar forms. On $L^2(M;\varphi)$ we can also define the modular conjugation:
\beq
	J_{\varphi}: L^2(M;\varphi) \rightarrow L^2(M;\varphi), \qquad J_{\varphi}(\eta_{\varphi}(m)) = \eta_{\varphi}(m^*), 
\eeq
which is an antilinear isometry implementing the involution at the Hilbert space level. The modular conjugation will be of use to us later on. 

The question of the existence of operator valued weights was answered by Haagerup \cite{haagerup1979operator,haagerup1979operator2} and later refined by Falcone and Takesaki \cite{falcone1999operator}. The central result is that, given an inclusion $i: N \hookrightarrow M$, there exists a faithful, semifinite, normal operator valued weight $T \in P(M,N)$ if and only if there exist faithful, semifinite, normal weights $\varphi \in P(M)$ and $\varphi_0 \in P(N)$ such that
\beq \label{Existence of OVW}
	\sigma^{\varphi}_t \circ i = i \circ \sigma^{\varphi_0}_t, \; \forall t \in \mathbb{R}. 
\eeq
In words, Eq.~\eqref{Existence of OVW} states that an operator valued weight from $M$ to $N$ exists if and only if $N$ is preserved under the modular automorphism of weights on $M$. As has been investigated by Takesaki in Ref.~\cite{TAKESAKI1972306}, this is a rather stringent requirement which leads to a large degree of factorization in the larger algebra $M$. In the event that Eq.~\eqref{Existence of OVW} holds there will always exist an operator valued weight $T \in P(M,N)$ such that $\varphi = \varphi_0 \circ T$. 

\subsection{Index Theory}

Given an inclusion $i: N \hookrightarrow M$ and a representation $\pi: M \rightarrow B(H)$, we can define a dual inclusion $i': \pi(M)' \hookrightarrow \pi(N)'$. If $P(M,N)$ is non-empty, there exists a correspondence between $P(M,N)$ and $P(\pi(N)',\pi(M)')$ first introduced by Kosaki \cite{kosaki1998type}. Given an operator valued weight $T \in P(M,N)$, there exists a unique operator valued weight $T^{-1} \in P(\pi(N)',\pi(M)')$ such that
\beq
	\frac{d \Psi \circ T^{-1}}{d \varphi} = \frac{d \Psi}{d \varphi \circ T}, \qquad \forall \Psi \in P(\pi(M)'), \; \varphi \in P(N).
\eeq
We refer to $T^{-1}$ as the Kosaki dual of $T$. 

Suppose that $E \in C(M,N)$ is a conditional expectation. If $M$ and $N$ are factor algebras it can be shown that
\beq
	E^{-1}(\mathbb{1}) = \text{Ind}(E) \mathbb{1},
\eeq
where $\text{Ind}(E) \in \mathbb{R}_+$ is a number called the index of $E$ \cite{jones1983index,pimsner1986entropy}. If $C(M,N)$ is non-empty, there will always exist a conditional expectation $E_{min}$ whose index is minimal. The index of the minimal conditional expectation is denoted by
\beq
	[M:N] \equiv \text{Ind}(E_{min}),
\eeq 
and referred to as the index of the inclusion $i: N \hookrightarrow M$. If the set $C(M,N)$ is empty, we say that $[M:N]$ is infinite. If $E \in C(M,N)$ is a conditional expectation with finite index, we can normalize its Kosaki dual to obtain a conditional expectation $\overline{E} \equiv \text{Ind}(E)^{-1} E^{-1} \in C(\pi(N)',\pi(M)')$. 

Finally, to any inclusion $i: N \hookrightarrow M$ we can associate an endomorphism $\gamma: M \rightarrow M$ called the canonical endomorphism \cite{Longo:1989tt,longo1995nets}. Let $\varphi \in P(M)$ be a weight whose restriction $\varphi_0 \circ i$ remains faithful, semifinite and normal. Then, the vector representative $\xi_{\varphi} \in L^2(M;\varphi)$ is cyclic and separating for both $\pi_{\varphi}(M)$ and $\pi_{\varphi} \circ i(N)$. Let $J_M$ denote the modular conjugation induced by $\varphi$ and $J_N$ denote the modular conjugation induced by $\varphi_0$ but acting on the Hilbert space $L^2(M;\varphi)$. The canonical endomorphism is defined:
\beq
	\gamma: M \rightarrow N, \qquad \gamma(m) \equiv \pi_{\varphi}^{-1} \circ \text{Ad}_{J_N J_M} \circ \pi_{\varphi}(m). 
\eeq
The image of the map $\gamma$ lies inside of $N$, and thus $\gamma$ may be regarded as a homomorphism between $M$ and $N$. As we shall see, this endomorphism plays a very central role in the construction of the Q-system. 
 
\subsection{Crossed Products as a Case Study} \label{sec: CP primer}

In this note, our objective is to understand how one can build extensions of an algebra $N$ by utilizing structures that only make reference to $N$. Our paradigmatic example is the crossed product of a von Neumann algebra by a locally compact group. In the main text, we show that general inclusions share many of the same features as the crossed product.  

A triple $(N,G,\alpha)$ consisting of a von Neumann algebra $N$, a locally compact group $G$ and an automorphic action $\alpha: G \rightarrow \text{Aut}(N)$ is called a von Neumann covariant system. Given any representation $\pi: N \rightarrow B(H)$, we can construct a covariant representation of $(N,G,\alpha)$ by extending $H \mapsto H_G \equiv L^2(G,H) \simeq H \otimes L^2(G)$ and defining
\beq
	\bigg(\pi_{\alpha}(n) \xi\bigg)(g) \equiv \pi \circ \alpha_{g^{-1}}(n) \bigg(\xi(g)\bigg), \qquad \bigg(\lambda(h) \xi\bigg)(g) \equiv \xi(h^{-1} g). 
\eeq
Here, $\pi_{\alpha}: N \rightarrow B(H_G)$ and $\lambda: G \rightarrow U(H_G)$ are a representation of $N$ and a unitary representation of $G$, respectively. They are compatible with the automorphism $\alpha$ in the sense that
\beq
	\pi_{\alpha} \circ \alpha_g(n) = \lambda(g) \pi_{\alpha}(n) \lambda(g^{-1}). 
\eeq
This establishes that the algebra generated by $\pi_{\alpha}(n) \lambda(g)$ is closed inside of $B(H_G)$ and can therefore be closed in the weak operator topology to induce a von Neumann algebra
\beq
	N \times_{\alpha} G \equiv \pi_{\alpha}(N) \vee \lambda(G)
\eeq
called the crossed product of $N$ by $G$. 

The algebra $N$ can be regarded as embedded in the crossed product via the map $\pi_{\alpha}: N \rightarrow M \equiv N \times_{\alpha} G$, which we view as an inclusion. It was shown by Digernes and Haagerup \cite{digernes1975duality,haagerup1978dual,haagerup1978dual2} that there will always exist an operator valued weight $T: M \rightarrow \hat{N}$. This weight can be constructed by first recognizing that a generic operator in $M$ may be written as (the limit of operators of the form)
\beq \label{CP Basis}
	\mathfrak{X} = \int_{G} d\mu(g) \; \lambda(g) \pi_{\alpha}(\mathfrak{X}(g)),
\eeq
where here $\mathfrak{X}: G \rightarrow N$. Then, Haagerup's operator valued weight is given by
\beq
	T(\mathfrak{X}^* \mathfrak{X}) = \int_G d\mu(g) \; \mathfrak{X}(g)^* \mathfrak{X}(g). 
\eeq
If $G$ is a finite group, then $T$ can be normalized to a conditional expectation. Its index is equal to the order of the group $|G|$, and thus gives rise to a dual conditional expectation, as well. If $G$ is not finite, $T$ is merely an operator valued weight and cannot be normalized. 
\section{Categorification of von Neumann Algebras} \label{app: Cat vN}

In this appendix we review a categorical point of view on von Neumann algebras largely following the presentation of \cite{Bischoff:2014xea}. Recall that a category is a set consisting of objects and arrows which are maps between objects. If, moreover, there are arrows which map between arrows, we recover the structure of a higher category. In this case, arrows between objects are called $1$-arrows, arrows between $1$-arrows are called $2$-arrows and so on and so forth. A higher category with $n$-arrows is called an $n$-category. For our purpose it will be useful to define a $2$-category associated with the set of von Neumann algebras. 

\subsection{Definition} 

Given a pair of von Neumann algebras $N$ and $M$, a homomorphism is a linear unital map $\alpha: N \rightarrow M$ preserving the product and involution
\beq
	\alpha(nn') = \alpha(n)\alpha(n'), \qquad \alpha(n^*) = \alpha(n)^*.
\eeq	
Given a pair of homomorphisms $\alpha,\beta: N \rightarrow M$ an \emph{intertwiner} of $\alpha$ and $\beta$ is an element $t \in im(\alpha) = im(\beta)$ such that
\beq
	t \alpha(n) = \beta(n) t, \; \forall n \in N. 
\eeq
For a given pair $\alpha, \beta$ we denote the set of intertwiners by $\text{Hom}(\alpha,\beta)$. If $\alpha,\beta,\gamma: N \rightarrow M$ and $t \in \text{Hom}(\alpha,\beta), s \in \text{Hom}(\beta,\gamma)$ then we can show that $st \in \text{Hom}(\alpha,\gamma)$:
\beq
	st \alpha(n) = s \beta(n) t = \gamma(n) st. 
\eeq
This defines a notion of composability for intertwiners. Using the aforementioned structures, we can now introduce the following category:

\begin{definition}[von Neumann category]
	The category of von Neumann algebras, vN, has as objects von Neumann algebras, as $1$-arrows homomorphisms between von Neumann algebras and as $2$-arrows intertwiners of homomorphisms. For a fixed von Neumann algebras the set of endomorphisms $\alpha: N \rightarrow N$ has the structure of a category denoted by $\text{End}(N)$. We can regard $\text{End}(N) \subset \text{vN}$ as a subcategory. 
\end{definition}
Notice that $\text{Hom}(\alpha,\beta) \subset im(\alpha)$ and thus inherits its norm and weak operator topology. In general $\text{Hom}(\alpha,\beta)$ is a complex vector space, but if $\alpha = \beta$ then $\text{Hom}(\alpha,\alpha)$ is in fact a $C^*$ algebra. For this reason, the category $\text{End}(N)$ is referred to as a $C^*$ category. 

\subsection{Monoidal Product}

The set vN is moreover a strict tensor category, which means that it admits a monoidal product $\otimes: \text{vN} \times \text{vN} \rightarrow \text{vN}$ which is strictly associative. The monoidal product on von Neumann algebras is simply the tensor product. The monoidal product on homomorphisms is the composition:
\beq
	\alpha: N \rightarrow Q, \beta: Q \rightarrow M, \qquad \beta \otimes \alpha \equiv \beta \circ \alpha: N \rightarrow M.
\eeq
Let $\alpha_i,\beta_i: N_i \rightarrow M_i$ with $M_2 = N_1$. Moreover, let $t_i \in \text{Hom}(\alpha_i,\beta_i) \subset im(\alpha_i) = M_i$. The monoidal product of intertwiners is given by
\beq
	t_1 \otimes t_2 \equiv t_1 \alpha_1(t_2) = \beta_1(t_2) t_1 \in \text{Hom}(\alpha_1 \otimes \alpha_2, \beta_1 \otimes \beta_2). 
\eeq	
A simple computation assures this is true:
\begin{flalign}
	t_1 \otimes t_2 \alpha_1 \otimes \alpha_2(n_2) &= t_1 \alpha_1(t_2) \alpha_1 \circ \alpha_2(n_2) \nonumber \\
	&= t_1 \alpha_1(t_2 \alpha_2(n_2)) \nonumber \\
	&= \beta_1(\beta_2(n_2) t_2) t_1 \nonumber \\
	&= \beta_1 \circ \beta_2(n_2) \beta_1(t_2) t_1 = \beta_1 \otimes \beta_2(n_2) t_1 \otimes t_2. 
\end{flalign}

Let $\alpha \in \text{End}(M)$ and denote by $\mathbb{1}_{\alpha} = \mathbb{1}_{im(\alpha)} \in im(\alpha)$. Clearly $\mathbb{1}_{\alpha}$ is a trivial intertwiner in $\text{Hom}(\alpha,\alpha)$. Given $\beta: M \rightarrow N$ and $\gamma: M \rightarrow N$ and an intertwiner $t \in \text{Hom}(\beta,\gamma)$ we have
\beq \label{otimes 1}
	t \otimes \mathbb{1}_{\alpha} = t \beta(\mathbb{1}_{\alpha}) = \gamma(\mathbb{1}_{\alpha}) t = t \in \text{Hom}(\beta \otimes \alpha, \gamma \otimes \alpha) \subset im(\beta \circ \alpha) = im(\beta). 
\eeq
Conversely, if $\beta: N \rightarrow M$ and $\gamma: N \rightarrow M$ and $t \in \text{Hom}(\beta,\gamma) \subset im(\beta)$ then
\beq \label{1 otimes}
	\mathbb{1}_{\alpha} \otimes t = \mathbb{1}_{\alpha} \alpha(t) = \alpha(t) \in \text{Hom}(\alpha \otimes \beta, \alpha \otimes \gamma). 
\eeq
\subsection{Conjugate Homomorphisms}

Given $\alpha: N \rightarrow M$ a homomorphism $\overline{\alpha}: M \rightarrow N$ is said to be conjugate if there exist a pair of intertwiners $w \in \Hom{id_N}{\overline{\alpha} \circ \alpha} \subset N$ and $\overline{w} \in \Hom{id_M}{\alpha \circ \overline{\alpha}} \subset M$ such that
\beq
	(w^* \otimes \mathbb{1}_{\overline{\alpha}})(\mathbb{1}_{\overline{\alpha}} \otimes \overline{w}) = \mathbb{1}_{\overline{\alpha}}, \qquad (\mathbb{1}_{\alpha} \otimes w^*)(\overline{w} \otimes \mathbb{1}_{\alpha}) = \mathbb{1}_{\alpha}. 
\eeq	
From \eqref{otimes 1} and \eqref{1 otimes} we can write these as
\beq
	w^* \overline{\alpha}(\overline{w}) = \mathbb{1}_N, \qquad \alpha(w)^* \overline{w} = \mathbb{1}_M. 
\eeq
We call the pair $(w,\overline{w})$ conjugate intertwiners for $(\alpha,\overline{\alpha})$. Since $(w, \overline{w})$ are intertwiners of the identity homomorphism there exist positive real (potentially infinite) scalars $d, \overline{d}$ such that
\beq
	w^*w = d \mathbb{1}_N, \qquad \overline{w}^* \overline{w} = \overline{d} \mathbb{1}_M. 
\eeq
One can always choose a normalization for $(w,\overline{w})$ such that $d = \overline{d} \equiv d(w,\overline{w})$, which we call the dimension of the conjugate pair. We define the dimension of the homomorphism $\alpha$ to be the infimum over the dimensions for all conjugate pairs:
\beq
	\text{dim}(\alpha) = \text{dim}(\overline{\alpha}) \equiv \inf_{(w,\overline{w})} d(w,\overline{w}). 
\eeq
A conjugate pair for which $d(w,\overline{w}) = \text{dim}(\alpha)$ is called standard. 

A homomorphism can be seen to be an isomorphism if and only if $\text{dim}(\alpha) = 1$. In general $\alpha(N) \subset M$ and the dimension of a homomorphism is equivalent to the square root of the index:
\beq
	\text{dim}(\alpha) = [M : \alpha(N)]^{1/2}. 
\eeq
In particular, if $N \subset M$ is a subfactor and $i: N \hookrightarrow M$ is its inclusion, then
\beq
	\text{dim}(i) = [M:N]^{1/2}. 
\eeq
We note that the map $i \circ \overline{i} \in \text{End}(M)$ is called the canonical endomorphism of $N \subset M$, while $\overline{i} \circ i \in \text{End}(N)$ is called the dual canonical endomorphism. The dimension is multiplicative over composition and additive over direct sum:
\beq
	\text{dim}(\alpha \circ \beta) = \text{dim}(\alpha) \text{dim}(\beta), \qquad \text{dim}\bigg( \bigoplus_{i \in \mathcal{I}} \alpha_i \bigg) = \sum_{i \in \mathcal{I}} \text{dim}(\alpha_i). 
\eeq

\providecommand{\href}[2]{#2}\begingroup\raggedright\endgroup

\clearpage

\end{document}